\documentclass[letterpaper, 10 pt, conference]{ieeeconf}  % Comment this line out if you need a4paper

\IEEEoverridecommandlockouts                              % This command is only needed if
\usepackage{despaper}

\title{\LARGE \bf
Distributed Secret Securing in Discrete-Event Systems
}

\author{Shoma Matsui$^{1}$, Kai Cai$^{2}$, and Karen Rudie$^{1}$% <-this % stops a space
\thanks{$^{1}$Department of Electrical and Computer Engineering and Ingenuity Labs Research Institute, Queen's University, Kingston, Canada. S. Matsui: \texttt{s.matsui@queensu.ca}, K. Rudie: \texttt{karen.rudie@queensu.ca}}%
\thanks{$^{2}$Department of Core Informatics, Osaka Metropolitan University, Osaka, Japan. \texttt{cai@omu.ac.jp}}%
}

\begin{document}
\maketitle
\thispagestyle{empty}
\pagestyle{empty}

\bstctlcite{BSTcontrol}

%%%%%%%%%%%%%%%%%%%%%%%%%%%%%%%%%%%%%%%%%%%%%%%%%%%%%%%%%%%%%%%%%%%%%%%%%%%%%%%%
\begin{abstract}
In this paper, we study a security problem of protecting secrets in distributed systems. Specifically, we employ discrete-event systems to describe the structure and behaviour of distributed systems, in which global secret information is separated into pieces and stored in local component agents. The goal is to prevent such secrets from being exposed to intruders by imposing appropriate protection measures.
This problem is formulated as to ensure that at least one piece of every distributed global secret is secured by a required number of protections, while the overall cost to apply protections is minimum. We first characterize the solvability of this security problem by providing a necessary and sufficient condition, and then develop an algorithm to compute a solution based on the supervisory control theory of discrete-event systems. Finally, we illustrate the effectiveness of our solution with an example system comprising distributed databases.
\end{abstract}

%%%%%%%%%%%%%%%%%%%%%%%%%%%%%%%%%%%%%%%%%%%%%%%%%%%%%%%%%%%%%%%%%%%%%%%%%%%%%%%%
\section{Introduction}\label{sec:intro}

As more and more information is stored on systems exposed to the Internet, it has been indispensable for system administrators to install and set up appropriate protection measures in their systems. While various practices and software applications have been developed to protect systems against intruders~\cite{Scarfone_2008}, it is still arduous even for skilled system administrators to properly set up such technologies especially when they have complicated systems and limited budgets. Along with evolving distributed systems and infrastructures that are typically too large and complex for system administrators to protect in an ad-hoc manner, formal methods have been attracting much attention from both research and industrial communities for system security~\cite{Kulik_2022}.

This paper employs discrete-event systems (DES)~\cite{Cassandras_2021} to model the system's structure and explain its behaviour. Our methodology also utilizes the central techniques from supervisory control theory (SCT)~\cite{Wonham_2019} of DES. Although DES and SCT were originally proposed as useful tools to address system control problems, subsequent works have attempted to apply DES and SCT to various security problems. For example, \emph{opacity} is one of the security notions in DES that focuses on anonymity and secrecy and has been extensively studied thus far~\cite{Bryans_2005,saboori2007notions,lin2011opacity,Balun_2022,Labed_2022,Han_2023}. In particular, current-state opacity (CSO) first introduced in~\cite{Bryans_2005} considers that the system of interest holds secret information modelled as \emph{states} of a DES that should not be exposed to adversarial entities. This paper adopts the same concept in a distributed setting.
There are also works that have examined how to enforce opacity through control \cite{dubreil2008opacity,Saboori_2012}, but our focus here is not on enacting control to ensure opacity but rather how to trade off hiding events to ensure some degree of opacity while at the same time leaving enough events not hidden so that they can be used to meet the protection requirements.
Other concepts of well-studied DES security are \emph{sensor deception}~\cite{Wakaiki_2018} and \emph{actuator enablement}~\cite{Carvalho_2016} attacks that consider more direct attacks against observation and control layers of DES than opacity and this paper as well as how to effectively defend systems against such attacks.

Security of non-monolithic systems has recently been attractive for the DES community. Such systems are often called decentralized or modular systems. For example, the decentralized version of language-based opacity was introduced in \cite{Paoli_2012}. Later on, the decentralized counterpart of current-state opacity was proposed by \cite{Wu_2018} and \cite{Tong_2018}. These works of decentralized settings differ from this paper, since they considered systems consisting of one agent. On the other hand, we consider distributed systems comprising more than one agent. Notably, the work of \cite{Tong_2019} which addressed so-called \emph{modular} opacity considered a similar distributed structure to this paper. One of the significant differences of this paper from \cite{Tong_2019} is that our methodology does not involve parallel composition whose computational complexity is exponential. That is, our formulated problem only needs to focus on each local agent rather than on a global observation.

Previous work \cite{Matsui_2019} introduced the notion of multiple protections to secure secrets in centralized DES. Specifically, we represent secrets as particular states in the system,
and from the security viewpoint, our system model has two types of events: \emph{protectable events} and \emph{unprotectable events}. Protectable events represent operations or software applications to which system administrators can apply protections. On the other hand, unprotectable events indicate the parts that system administrators cannot protect due to technical difficulty or resource deficiency. Along with the partition of protectable and unprotectable events, we also consider that each protectable event is associated with a ``cost'' level that represents the level of technical or financial burden for system administrators.
The goal was to find which transitions in the system should be ``protected'', namely a \emph{protection policy}, so that every trajectory from the initial state to the secret states has at least the required number of protected transitions. Protected transitions here represent restricted operations or access control on the real systems.

The protection policy also needs to only yield necessary transitions to protect, so that the overall cost to impose the policy on the system is minimum. In \cite{Matsui_2019}, a special case of this problem was formulated as \emph{Secret Securing with Two Protections and Minimum Costs Problem} where two protections are required for every secret (generalizable to $r~(\geq 2)$ protections). To find a solution of this problem, an algorithm that employs SCT was developed by converting the notion of protection to that of control. In other words, protecting transitions was interpreted as disabling them in the context of SCT.

The subsequent work of \cite{Matsui_2021} extended \cite{Matsui_2019} to heterogenous secret importance, that is, secrets in the system may have different levels of importance. Specifically in \cite{Matsui_2021}, we are given a security requirement that specifies the minimum cost level of protections for each group of secrets in addition to the number of protections. This assumption represents that secrets associated with a higher importance level should be protected by more costly protections, since the cost levels of protections also model the difficulty levels for intruders to breach in general. Another notion introduced in \cite{Matsui_2021} was so-called \emph{usability}, i.e., how much protecting transitions negatively affects legitimate users, since they are forced to pass through applied protections as well as intruders. A similar notion, called \emph{disruptiveness}, was also studied in \cite{Ma_2021} and \cite{Ma_2023}. The main difference is that in \cite{Matsui_2021}, negative effects on usability were captured by incrementing the cost levels of events, while for disruptiveness, minimal disruption can be achieved by assigning the same cost to all events and protecting as few events as possible.

In this paper, we further extend the earlier works of \cite{Matsui_2019} and \cite{Matsui_2021} from monolithic systems to distributed (i.e., multiagent) ones. Simply put, we consider a system modelled as a DES that consists of more than one agent, and each local agent holds ``pieces'' of global secret information that are distributed across agents. Given the structures of agents, the partition of protectable and unprotectable events, and the classes of cost levels, our goal in this paper is to find a protection policy such that intruders can never retrieve all pieces from local agents so that global secrets are protected, and the overall cost to impose the policy is minimum. This problem setting is motivated by distributed databases where secret information is divided into pieces and stored in multiple databases across several agents. We will illustrate a motivating example later in \cref{exmp:motivation}.

The remainder of this paper is organized as follows. In \cref{sec:dss}, we introduce our system model represented as a DES and formulate our goal of finding a protection policy that satisfies specific conditions. \cref{sec:synthesis} first presents a necessary and sufficient condition under which our problem is solvable, and then proposes an algorithm to construct a solution of our problem if one exists. In \cref{sec:example} we demonstrate the algorithm with a running example, and \cref{sec:conclusion} concludes this paper.

% \section{Preliminaries}\label{sec:prelim}

\section{Distributed Secret Securing}\label{sec:dss}

\subsection{System Models}\label{sec:dss:models}

We consider distributed systems consisting $n$ local agents. Each agent is modelled as an automaton denoted by a 4-tuple:
\begin{equation}\label{eq:agent}
  G^i = (Q^i, \Sigma^i, \delta^i, q_0^i)
\end{equation}
where $i \in [1, n]$ is the index of agents, $Q^i$ is the set of states, $\Sigma^i$ is the set of events, $\delta^i$ is the (partial) transition function, and $q_0^i \in Q^i$ is the initial state. The set of events $\Sigma^i$ is partitioned into two disjoint subsets: protectable event set $\Sigma_p^i$ and unprotectable event set $\Sigma_{up}^i$, namely $\Sigma^i = \Sigma_p^i \disjoint \Sigma_{up}^i$.

The set of all protectable events, given by $\Sigma_p = \bigcup_{i=1}^{n} \Sigma_p^i$, is further partitioned into $m$ disjoint subsets to represent the classes of cost levels, i.e., $\Sigma_p = \bigdisjoint_{j=1}^m \Sigma_{cl,j}$. Observe the difference of the composition of $\Sigma_p$, i.e., the union of $\Sigma_p^i$ and the disjoint union of $\Sigma_{cl,j}$. This indicates that agents could share protectable events while no protectable event belongs to more than one cost class.

\begin{rem}
  The cost level of each class is not comparable with other classes. In other words, the cost to protect one event in $\Sigma_{cl,2}$ is sufficiently higher than the cost to protect all events in $\Sigma_{cl,1}$.
\end{rem}

To discuss secret securing for systems with multiple agents defined above, we first adopt the notion of \emph{secure reachability} from previous work~\cite{Matsui_2019} with a slight adjustment.

\begin{defn}[$r$-secure Reachability]\label{defn:secure}
  Consider an agent $G^i$ in \cref{eq:agent}. Given a required protection $r$, a secret state $q_s \in Q_s^i$ is said to be \emph{securely reachable} with at least $r$ protectable events ($r$-securely reachable) with respect to $G$ and $\Sigma_{p,k} = \bigdisjoint_{i=1}^k \Sigma_{cl,i}$ if
  \[
    (\forall s \in \Sigma^*)\ \delta(q_0^i, s) = q_s \implies s \in \underbrace{\Sigma^* \Sigma_{p,k} \Sigma^* \dots \Sigma^* \Sigma_{p,k} \Sigma^*}_{\text{$\Sigma_{p,k}$ appears $r$ times}}
  \]
\end{defn}

In words, a secret state $q_s$ is securely reachable if the intruder has to break at least $r$ protections before they reach $q_s$ from the initial state $q_0^i$.

% change S to a set of "global secrets" and one tuple will be "a global secret".

We suppose in this paper that secret information is distributed across all agents and each agent holds pieces of secrets. Specifically, each secret is represented as a tuple of secret states from all agents. For example, when we store a password ``abcd1234'' in a distributed system consisting of two database servers, one possible way is to split the password into two pieces, such as ``abcd'' and ``1234''. We call a tuple of secret states a \emph{global secret}, and denote the set of global secrets by $S$, namely the set of $n$-tuples of secret states from all agents:
\begin{equation}\label{eq:secret}
  S \subseteq Q^1_s \times Q^2_s \times \dots \times Q^n_s.
\end{equation}
We refer to the elements inside each $Q^i_s$ as \emph{local secret states}.
Note that different tuples in $S$ can share the same local secret states. For example, let $Q^1_s = \paren{q^1_a, q^1_b}$ and $Q^2_s = \paren{q^2_a, q^2_b}$. Since $Q^1_s \times Q^2_s = \paren{(q^1_a, q^2_a), (q^1_b, q^2_a), (q^1_a, q^2_b), (q^1_b, q^2_b)}$, $S$ may be $\paren{(q^1_a, q^2_a), (q^1_a, q^2_b)}$, i.e., a local secret state $q^1_a$ is shared by two tuples in $S$. This models our motivating example of distributed passwords. For instance, two distinct passwords ``abcdefgh'' and ``abcd1234'' shares the same piece of a local secret ``abcd''.

Besides the global secret, we consider that each secret tuple needs to be secured by a prescribed number of protections. For instance, it is natural that the more important a secret is, the more protections is needed. In other words, one would want to use fewer protections to reduce costs if a secret is not as important as others. We denote such a requirement of the number of protections for each secret in $S$ by
\begin{equation}\label{eq:requirement}
  R = \{(s, r) \mid s \in S, r \in \mathbb{N}\}
\end{equation}
where $\mathbb{N}$ is the set of all integers greater than $0$. We henceforth call $R$ in \cref{eq:requirement} the \emph{security requirement}. Also for convenience, we denote by $s^i$ the $i$\textsuperscript{th} component of $s$.

\subsection{Problem Formulation}\label{sec:dss:problem}

In this paper, we consider that intruders need to retrieve all pieces of a global secret in order to obtain this secret. Thus, the goal of system defenders is to have at least one state of every tuple in the global secret set secured by the corresponding number of protections, represented as $R$. We define the notion of a secured set of global secrets as follows.

\begin{defn}[Secured Set of Global Secrets]\label{defn:security}
  Consider an agent $G^i$ in \cref{eq:agent}, a set of global secrets $S$ in \cref{eq:secret}, and a security requirement $R$ in \cref{eq:requirement}. The set of global secrets $S$ is said to be \emph{secure} with respect to $R$ and $\Sigma_{p,k}$ if for all $(s, r) \in R$, there exists $k' \in [1, m]$ such that some $s^i$ is $r$-securely reachable with respect to $G^i$ and $\Sigma_{p,k'}$, and $k = \max_R k'$.
  % On the other hand, $S$ is said to be \emph{strongly secure} with respect to $R$ and $\Sigma_{p,k}$ if for all $(s, r) \in R$, there exists $k' \in [1,m]$ such that for all $s^i \in s$, $s^i$ is $r$-securely reachable with respect to $G^i$ and $\Sigma_{p,k'}$ where $k = \max_R k'$.
\end{defn}

By taking the maximum of $k'$ among all elements of $R$, the $k$ in \cref{defn:security} represents the maximum class of cost levels among events that should be protected to satisfy the security requirement $R$.
% The difference of the strong version is that all states of the global secret must satisfy the security requirement whilst the non-strong version imposes the security requirement upon at least one state in every tuple of the global secret. Thus, strong secret security indeed implies secret security.

In addition to \cref{defn:secure}, we also employ from \cite{Matsui_2019} the notion of \emph{protection policy} that describes which transitions should be protected.

\begin{defn}[Protection Policy]\label{defn:policy}
  For an agent modelled as $G^i = (Q^i, \Sigma^i, \delta^i, q_0^i)$ and a nonempty subset of protectable events $\Sigma_p^i \subseteq \Sigma^i$, a \emph{protection policy} $\mathcal{P}^i: Q^i \to 2^{\Sigma_p^i}$ is a mapping function that assigns to each state a subset of protectable events.
\end{defn}

Using \cref{defn:secure}, \cref{defn:security} and \cref{defn:policy}, we now formulate the defender's goal of protecting the global secret against intruders with minimum costs.

\begin{prob}[Distributed Secret Securing Problem, DSSP]\label{prob:dssp}
  Consider agents modelled as in \cref{eq:agent}, a set of global secrets $S$ in \cref{eq:secret}, and a security requirement $R$ in \cref{eq:requirement}. Find protection policies $\mathcal{P}^i$ of $G^i$, $i \in [1,n]$, such that
  \begin{enumerate}[label={P\arabic*.},ref={P\arabic*}]
    \item \label{cond:dssp:secure} $S$ is secure with respect to $R$ and $\Sigma_{p,k}$; and
    \item \label{cond:dssp:minimum} $k$ is minimum.
  \end{enumerate}
\end{prob}

The conditions \ref{cond:dssp:secure} and \ref{cond:dssp:minimum} require the security requirement to be satisfied by minimum costs, that is, \ref{cond:dssp:secure} does not hold for $k-1$.

\begin{exmp}\label{exmp:motivation}
  \begin{figure}[htp]
    \centering
    \begin{subfigure}{\linewidth}
      \centering
      \begin{adjustbox}{scale=0.9}
        \begin{tikzpicture}[automata]
          \node[state,initial] (1) {$q_0^1$};
          \node[state,above right=1.5 and 1 of 1] (2) {$q_1^1$};
          \node[state,below right=1.5 and 1 of 2, secret] (3) {$q_2^1$};

          \draw (1) -- node{$\alpha_1$} (2);
          \draw (2) -- node{$\alpha_2$} (3);
          \draw (3) -- node{$\alpha_3$} (1);
        \end{tikzpicture}
      \end{adjustbox}
      \caption{$G^1 = (Q^1, \Sigma^1, \delta^1, q_0^1)$}
    \end{subfigure} \\ \bigskip
    \begin{subfigure}{\linewidth}
      \centering
      \begin{adjustbox}{scale=0.9}
        \begin{tikzpicture}[automata]
          \node[state] (1) {$q_0^2$};
          \node[state,above=of 1] (2) {$q_1^2$};
          \node[state,left=of 1,secret] (3) {$q_2^2$};
          \node[state,below=of 1] (4) {$q_3^2$};
          \node[state,right=of 1,secret] (5) {$q_4^2$};

          \draw (-6mm,-6mm) -- (1);
          \draw (1) -- node{$\beta_1$} (2);
          \draw (2) -- node[swap]{$\beta_2$} (3);
          \draw (3) -- node[swap]{$\beta_3$} (4);
          \draw (1) -- node{$\beta_4$} (4);
          \draw (4) -- node[swap]{$\beta_5$} (5);
          \draw (5) -- node[swap]{$\beta_6$} (2);
          \draw (3) -- node[above]{$\beta_7$} (1);
          \draw (5) -- node[above]{$\beta_7$} (1);
        \end{tikzpicture}
      \end{adjustbox}
      \caption{$G^2 = (Q^2, \Sigma^2, \delta^2, q_0^2)$}\label{fig:exmp:motivation:2}
    \end{subfigure}
    \caption{Example agents}\label{fig:exmp:motivation}
  \end{figure}
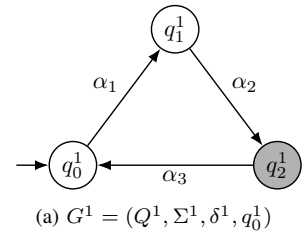
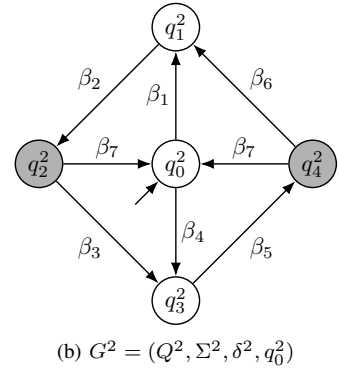
  Let us consider example agents in \cref{fig:exmp:motivation} which model a distributed database consisting of two servers: $G^1$ (Server 1) and $G^2$ (Server 2). Here, let us suppose that users store their passwords in this system, and each password is divided into two pieces. Moreover, $G^2$ consists of two databases, represented by secret states $q_2^2$ and $q_4^2$, and either is used to store pieces of passwords. On the other hand, $G^1$ has only one database indicated by a secret state $q_2^1$.

  In $G^1$, three states $q_0^1$, $q_1^1$, and $q_2^1$ represent results of specific operations by its users: $q^1_0$ (logged out), $q^1_1$ (logged in), and $q^1_2$ (accessed a database).
  Accordingly, the events in $G^1$ illustrate the following operations: $\alpha_1$ (logging into Server 1), $\alpha_2$ (accessing the database), and $\alpha_3$ (logging out from Server 1).
  In addition to the events and states, secrets and protections in $G^1$ are given by $Q^1_s = \paren{q^1_2}$ and $\Sigma^1_p = \paren{\alpha_1, \alpha_2}$.

  On the other hand, the states and events in $G^2$ are described as follows:
  \begin{itemize}
    \item $q^2_0$: logged out
    \item $q^2_1$: logged in as User 1
    \item $q^2_2$: accessed Database 1
    \item $q^2_3$: logged in as User 2
    \item $q^2_4$: accessed Database 2
    \item $\beta_1$: logging into the server as User 1
    \item $\beta_2$: accessing Database 1
    \item $\beta_3$: switching the account to User 2 from User 1
    \item $\beta_4$: logging in as User 2
    \item $\beta_5$: accessing Database 2
    \item $\beta_6$: switching the account to User 1 from User 2
    \item $\beta_7$: logging out.
  \end{itemize}
  This setup represents a situation where accessing different databases, namely Database 1 or Database 2, requires different permissions. That is, in order to access Database 1 (resp., Database 2), users must log into the server as User 1 (resp., User 2). Similar to $G^1$, $G^2$ has the following secret states and protectable events: $Q^2_s = \{q^2_2, q^2_4\}$ and $\Sigma^2_p = \{\beta_1, \beta_2, \beta_4, \beta_5\}$.

  Now with the setup of states and events, suppose that we have the following global secret and security requirement:
  \begin{align}
    S &= \{\underbrace{(q^1_2, q^2_2)}_{s_1}, \underbrace{(q^1_2, q^2_4)}_{s_2}\}, \label{eq:exmp:secret} \\
    R &= \paren{(s_1, 1), (s_2, 2)} \label{eq:exmp:req}
  \end{align}
  The set of global secrets $S$ in \cref{eq:exmp:secret} models that the user enters two passwords $s_1$ and $s_2$ (e.g., ``pass1234'' and ``pass5678'') and these two pairs share the same secret state $q^1_2$ which corresponds to the first half of two passwords (e.g., ``pass''). As seen in \cref{eq:exmp:req}, we are required to protect $s_1$ by one protection, i.e., we need to make sure that all users of this system, especially intruders, must pass at least one protected transition before reaching either $q^1_2$ or $q^2_2$. The reasoning behind this requirement is that intruders are supposed to reach both $q^1_2$ and $q^2_2$ in order to retrieve the full information of $s_1$, e.g., ``pass1234''. Also, as seen in \cref{eq:exmp:req}, we must protect $s_2$ by at least two protections.
  Moreover, consider that there are four classes of cost levels in this example, and the protectable events in $G^1$ and $G^2$ are grouped as follows:
  \begin{align*}
    \Sigma_{cl,1} &= \{\beta_1, \beta_4\} \\
    \Sigma_{cl,2} &= \{\alpha_1, \beta_2\} \\
    \Sigma_{cl,3} &= \{\alpha_2\} \\
    \Sigma_{cl,4} &= \{\beta_5\}.
  \end{align*}
  Our goal in this example is to find two control policies $\mathcal{P}^1$ and $\mathcal{P}^2$ for $G^1$ and $G^2$, respectively, so that the security requirement $R$ in \cref{eq:exmp:req} is satisfied with minimum costs.
\end{exmp}

\section{Synthesis of Protection Policies}\label{sec:synthesis}

% In this section, we first discuss the solvability of DSSP, namely the conditions where a solution of \cref{prob:dssp} exists, and then present an algorithm to synthesize a solution of DSSP.

\subsection{Solvability of DSSP}\label{sec:synthesis:solvability}

In this section, we present a necessary and sufficient condition under which a solution to \cref{prob:dssp} exists. The main idea is to verify that condition \ref{cond:dssp:secure} holds with $k = m$, namely the maximum $k$, for a given agent $G$ and security requirement $R$, i.e., the system is secure with respect to $R$ and $\Sigma_p$.

\begin{prop}\label{prop:solvability}
  Consider a set of global secrets $S$, a security requirement $R$, and a set of protectable events $\Sigma_p$. \cref{prob:dssp} is solvable if and only if the given $S$ is secure with respect to $R$ and $\Sigma_p$.
\end{prop}

\begin{proof}
  (if) Suppose that the given $S$ is secure with respect to $R$ and $\Sigma_p$. This means that $S$ is secure with respect to $R$ and $\Sigma_{p,m}$ according to \cref{defn:security}. We can construct protection policies that protect all transitions labelled by the protectable events in $\Sigma_p$. Such protection policies indeed satisfy the requirement $R$. Thus, we can conclude that there exists a solution of \cref{prob:dssp}, because even if $S$ is not secure with respect to $k$ smaller than $m$, policies protecting all protectable transitions can be a solution as \ref{cond:dssp:minimum} is satisfied with $k = m$.

  (only if) By \cref{defn:secure} and \cref{defn:security}, if the set of global secrets $S$ is not secure with respect to $R$ and $\Sigma_p = \Sigma_{p,m}$, then $S$ is not secure for all $k < m$. Namely, there does not exist any solution of \cref{prob:dssp} in that case. Thus, if \cref{prob:dssp} is solvable, then $S$ is secure with respect to $R$ and $\Sigma_p$.
\end{proof}

Note that \cref{prop:solvability} only states the existence of a solution to \cref{prob:dssp}, because a protection policy with $k = m$ may not satisfy the condition \ref{cond:dssp:minimum}.
% although \cref{prop:solvability} can be utilized to verify that \cref{prob:dssp} is solvable for given $G$, $S$, and $R$.
This naturally motivates us to develop an algorithm to compute a policy (if it exists) that satisfies both \ref{cond:dssp:secure} and \ref{cond:dssp:minimum}. In \cref{sec:synthesis:computation}, we present an algorithm by employing the supervisory control theory of DES.

\subsection{Policy Computation}\label{sec:synthesis:computation}

As done in \cite{Matsui_2019}, we convert the security problem \cref{prob:dssp} to a control problem so as to resort to SCT. In particular, the protection policy defined in \cref{defn:policy} can be directly derived from the \emph{control policy} of SCT by representing protection as disablement.

In SCT, a control policy $\mathcal{D}^i: Q^i \to 2^{\Sigma_p^i}$ for an agent $G^i$ represents the supervisor's control decision of which events to disable at any given state. Here, we slightly abuse the notation $\Sigma_p^i$ by taking it as the set of \emph{controllable events}, instead of protectable events.

\begin{defn}[Control Policy]\label{defn:policy:control}
  For an agent modelled as $G^i = (Q^i, \Sigma^i, \delta^i, q_0^i)$ and a nonempty subset of controllable events $\Sigma_p^i \subseteq \Sigma^i$, a control policy $\mathcal{D}^i: Q^i \to 2^{\Sigma_p^i}$ is a mapping function that assigns to each state a subset of controllable events. Specifically, letting $M^i = (X^i, \Sigma^i, \xi^i, x_0^i)$ be a supervisor that is a subautomaton of $G^i$, $\mathcal{D}^i$ is given by
  \begin{equation}\label{eq:defn:policy:control}
    \mathcal{D}^i(q) \coloneqq \begin{dcases}
      \{\sigma \in \Sigma_p^i \mid \neg\xi(q, \sigma)! \land \delta(q, \sigma)!\} & \text{if $q \in X^i$} \\
      \emptyset & \text{if $q \in Q^i \setminus X^i$}
    \end{dcases}
  \end{equation}
\end{defn}

It is straightforward to derive a protection policy $\mathcal{P}^i$ from a control policy $\mathcal{D}^i$, as they are mathematically the same mapping functions. Namely, we always have that
\[
  (\forall q \in Q^i)\ \mathcal{P}^i(q) = \mathcal{D}^i(q),
\]
thus we henceforth refer to protection policies and control policies interchangeably.

We emphasize that protecting transitions is a more appropriate representation than disablement, because administrators would not like to impose low availability on the target system, as disabling transitions prevents non-adversarial users from reaching secrets even if they are authorized to do so.

Besides the conversion from security to control, we employ the technique of relabelling transitions developed in \cite{Matsui_2019}. Roughly speaking, we ``relabel'' specific transitions originally labelled by controllable events. The main idea is that relabelled events are no longer controllable. This enables us to compute multiple supervisors for a single agent, exploiting that one supervisor computed by SCT specifies one transition in every trajectory from the initial state to designated states which are often called ``illegitimate'' states in the context of supervisory control problems. More precisely, relabelling is a procedure to construct a new automaton $G^{i\prime}$ based on $G^i = (Q^i, \Sigma^i, \delta^i, q_0^i)$ and $\mathcal{D}^i$.

First, we define $\delta_{\mathcal{D}^i}$ as a new transition function that indicates transitions specified by $\mathcal{D}^i$, i.e.,
\[
  \delta_{\mathcal{D}^i}(q, \sigma) \coloneqq \begin{dcases}
    q' & \text{if $q \in Q^i \land q' \in Q^i \land \sigma \in \mathcal{D}^i(q)$} \\
    \text{undefined} & \text{otherwise}
  \end{dcases}
\]
Next, we relabel $\delta_{\mathcal{D}^i}$ into $\delta_{\mathcal{D}^i}'$ by
\[
  % \delta_{\mathcal{D}}' \coloneqq \paren{(q, \sigma', q') \mid q, q' \in Q, \sigma \in \mathcal{D}(q)}.
  \delta_{\mathcal{D}^i}'(q, \sigma') \coloneqq \delta_{\mathcal{D}^i}(q, \sigma).
\]
Observe that we relabelled $\sigma$ of a transition specified by $\mathcal{D}^i$ into $\sigma'$, adding the superscript of prime. Finally, we construct a new automaton $G^{i\prime}$ by
\begin{equation}\label{eq:relabel}
  G^{i\prime} = (Q^i, \Sigma^{i\prime}, \delta^{i\prime}, q_0^i)
\end{equation}
where
\begin{align}
  \Sigma^{i\prime} &= \Sigma_{up}^{i\prime} \disjoint \Sigma_p^i \nonumber \\
  \Sigma_{up}^{i\prime} &= \Sigma_{up}^i \disjoint \{\sigma' \mid (\exists q \in Q^i)\ \delta_{\mathcal{D}}'(q, \sigma')!\} \nonumber \\
  \delta'(q, \sigma) &= \begin{dcases}
    q' & \text{if $\delta_{\mathcal{D}^i}'(q, \sigma) = q'\ \lor$} \\
    & \qquad \paren\big(\delta(q, \sigma) = q' \land \neg\delta_{\mathcal{D}^i}(q, \sigma)!) \\
    \text{undefined} & \text{otherwise}
  \end{dcases}\label{eq:relabel:remove}
\end{align}
Essentially, we replace the transitions in $G$ specified by $\mathcal{D}$ to relabelled ones as in \cref{eq:relabel:remove}.
\begin{figure}[htp]
  \centering
  \begin{adjustbox}{scale=1}
    \begin{tikzpicture}[automata]
      \node[state] (1) {$q_0^2$};
      \node[state,above=of 1] (2) {$q_1^2$};
      \node[state,left=of 1,secret] (3) {$q_2^2$};
      \node[state,below=of 1] (4) {$q_3^2$};
      \node[state,right=of 1,secret] (5) {$q_4^2$};

      \draw (-6mm,-6mm) -- (1);
      \draw[red] (1) -- node{$\beta_1'$} (2);
      \draw (2) -- node[swap]{$\beta_2$} (3);
      \draw (3) -- node[swap]{$\beta_3$} (4);
      \draw[red] (1) -- node{$\beta_4'$} (4);
      \draw (4) -- node[swap]{$\beta_5$} (5);
      \draw (5) -- node[swap]{$\beta_6$} (2);
      \draw (3) -- node[above]{$\beta_7$} (1);
      \draw (5) -- node[above]{$\beta_7$} (1);
    \end{tikzpicture}
  \end{adjustbox}
  \caption{Relabelled automaton; red arrows indicate relabelled transitions}\label{fig:exmp:relabelled}
\end{figure}
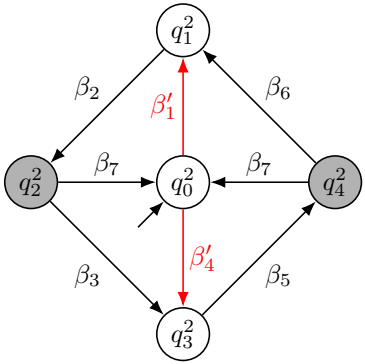
The above relabelling procedure is illustrated by \cref{fig:exmp:relabelled}, where the $G^i$ in question is that of \cref{fig:exmp:motivation}(b) and the policy $\mathcal{D}^i$ is given by
\begin{align*}
  \mathcal{D}^i(q_0^2) &= \paren{\beta_1, \beta_4} \\
  \mathcal{D}^i(q) &= \emptyset\ \text{for all $q \in q_1^2, q_2^2, q_3^2, q_4^2$}.
\end{align*}

\begin{algorithm}[tp]
  \caption{Distributed Reachability Control with Minimum Costs (DRCMC)}
  \label{alg:drcmc}
  Given agents $G^1$, $G^2$, \dots, $G^n$, and a security requirement $R$, this algorithm computes protection policies $\mathcal{P}^1$, $\mathcal{P}^2$, \dots, $\mathcal{P}^n$ for individual agents and the maximum cost level $k$ among computed policies, so that both conditions \ref{cond:dssp:secure} and \ref{cond:dssp:minimum} in \cref{prob:dssp} are satisfied.
  \begin{algorithmic}[1]
    \Require $G^1$, $G^2$, \dots, $G^n$, $R$
    \Ensure ($\mathcal{P}^1$, $\mathcal{P}^2$, \dots, $\mathcal{P}^n$, $k$) or \texttt{NO\_SOLUTION}
    \State $V = \emptyset$ \Comment{The set $V$ stores minimum cost levels determined in each iteration.}
    \State $\mathcal{P}^i(q) \coloneqq \emptyset$ for all $i \in [1,n]$ and $q \in Q^i$
    \For{$(s, r) \in R$} \label{ln:main:start}
      \State $W = \emptyset$ \Comment{The set $W$ stores pairs of a control policy and its maximum cost level.}
      \For{$s^i \in s$} \Comment{$s = (s^1, s^2, \dots, s^n)$} \label{ln:iter:mrcmc}
        \State $(\mathcal{D}^i, k^i) = \Call{MRCMC}{G^i, s^i, r}$
        \If{$\mathcal{D}^i \neq \text{Null}$} \label{ln:iter:null}
          \State $W = W \cup \paren{(\mathcal{D}^i, k^i)}$ \label{ln:iter:update}
        \EndIf
      \EndFor \label{ln:enditer:mrcmc}
      \If{$W = \emptyset$}
        \State \Return \texttt{NO\_SOLUTION} \label{ln:terminate}
      \EndIf
      \State $k_s = \max\paren*(\paren{k' \mid (\_, k') \in W})$ \label{ln:min:start}
      \State $j = 1$
      \For{$(\mathcal{D}^i, k^i) \in W$}\Comment{Find $i$ s.t. $k^i$ is minimum}
        \If{$k^i \leq k_s$}
        \State $k_s = k^i$
        \State $j = i$ \Comment{Save the candidate index} \label{ln:save}
        \EndIf
      \EndFor \label{ln:min:end}
      \State $\mathcal{P}^j(q) = \mathcal{P}^j(q) \cup \mathcal{D}^j(q)$ for all $q \in Q^j$ \label{ln:policy:protection}
      \State $V = V \cup \paren{k^j}$ \label{ln:costlevel}
    \EndFor \label{ln:main:end}
    \State $k = \max(V)$
    \State \Return ($\mathcal{P}^i$, $\mathcal{P}^2$, \dots, $\mathcal{P}^n$, $k$)
  \end{algorithmic}
\end{algorithm}

\begin{algorithm}[tp]
  \caption{Subroutines of DRCMC}
  \label{alg:sub}
  \begin{algorithmic}[1]
    \Procedure{MRCMC}{$G$, $q_s$, $r$} \Comment{This computes control policies that specify at least $r$ protectable transitions on every trajectory from the initial state to the secret state $q_s$ in $G$.}
    \State $T = \emptyset$
    \State $G_1 = G = (Q, \_, \_, \_)$
    \For{$j \in [1, r]$}
      \State $G_{K,j} = \Call{Spec}{G_j, q_s}$
      \State $(M_j, k_j) = \Call{RCMC}{G_j, G_{K,j}}$
      \If{$M_j \neq \text{Null}$}
        \State $T = T \cup \{k_j\}$
        \State Derive $\mathcal{D}_j$ from $M_j$ as in \cref{eq:defn:policy:control}
        \If{$j < r$}
          \State Form $G_{j+1}$ from $\mathcal{D}_j$ as in \cref{eq:relabel} \label{ln:mrcmc:relabel}
        \EndIf
      \Else
      \State \Return $(\text{Null}, 0)$
      \EndIf
    \EndFor
    \State $\mathcal{D}(q) \coloneqq \bigcup_{j = 1}^r \mathcal{D}_j(q)$ for all $q \in Q$
    \State \Return $\paren\big(\mathcal{D}, \max(T))$
    \EndProcedure
    \bigskip
    \Procedure{Spec}{$G$, $q_s$} \Comment{This removes from $G$ the secret state $q_s$ and transitions to and from $q_s$.}
      \State $\delta_K(q) = \begin{dcases*}
              \text{undefined} & \text{if $q = q_s \lor \delta(q) = q_s$} \\
              \delta(q) & \text{otherwise}
            \end{dcases*}$
      \State $G_K = (Q \setminus \paren{q_s}, \Sigma, \delta_K, q_0)$
      \State \Return $G_K$
    \EndProcedure
    \bigskip
    \Procedure{RCMC}{$G$, $G_K$} \Comment{This finds a nonempty supervisor (if one exists) and a maximum cost level which the supervisor needs to specify.}
      \State $K = L(G_K)$
      \For{$k = 1, 2, \dots, m$} \label{ln:rcmc:inc}
        \State $\Sigma_{p,k} = \bigdisjoint_{j = 1}^{k} \Sigma_{cl,j}$
        \State \label{ln:sub:sup} Compute a supervisor $M$ s.t. $L(M) = \supc(K)$ w.r.t. $\Sigma_{p,k}$
        \If{$M$ is nonempty}
          \State \Return $(M, k)$
        \EndIf
      \EndFor
      \State \Return $(\text{Null}, 0)$
    \EndProcedure
  \end{algorithmic}
\end{algorithm}

We are now ready to present \cref{alg:drcmc}, named \emph{Distributed Reachability Control with Minimum Costs} (DRCMC), which computes a solution of \cref{prob:dssp} if one exists.
The procedures in \cref{alg:sub} have been acquired from previous work \cite{Matsui_2019,Matsui_2021} with slight adjustments. \cref{alg:drcmc} takes distributed agents $G^i$ for $i \in [1,n]$ and their security requirement $R$, and returns protection policies of each agent and the minimum cost level if they exist, otherwise it terminates and returns \texttt{NO\_SOLUTION}.
First, \cref{alg:drcmc} initializes the protection policies and the temporary set $V$ with empty sets, explicitly indicating that no computation is invoked yet. As seen in line \ref{ln:costlevel}, the variable $V$ stores minimum cost levels determined in each iteration of the policy computation.
\cref{alg:drcmc} examines if every requirement of $R$ is satisfied by calling the procedure MRCMC. Specifically, MRCMC computes supervisors and corresponding control policies with respect to each of the given agents and their secret states. The crucial part of MRCMC is that due to the subroutine RCMC, MRCMC always returns a minimum cost level to satisfy the security requirement given as $r$.

After aggregating control policies for all secret states in the $n$-tuple $s$, DRCMC proceeds to find a policy whose cost is minimum compared to other states of $s$. If the set $W$ is empty, then \cref{alg:drcmc} returns a special value \texttt{NO\_SOLUTION} and terminates, meaning that no secret states in $s$ can be protected by at least $r$ protections. Once it finds a minimum cost and the corresponding control policy, DRCMC adds the cost level and the policy to $V$ and the initialized policy $\mathcal{P}^j$, respectively. The procedure of finding a minimum cost level and its corresponding policy is repeated until DRCMC examines all pairs in the security requirement $R$.

\begin{thm}\label{thm:drcmc}
  \cref{alg:drcmc} returns a solution of \cref{prob:dssp} if and only if \cref{prob:dssp} is solvable.
\end{thm}

\begin{proof}
  \cref{alg:drcmc} always terminates, since the security requirement $R$ and each tuple $s$ in $R$ are finite, and $k$ is bounded by $m$.

  The cost level $k$ returned by RCMC is minimum because it monotonically increments the cost level in line \ref{ln:rcmc:inc} of \cref{alg:sub}. Also, by the construction of the specification automaton in the procedure \textsc{Spec} and the relabelling in line \ref{ln:mrcmc:relabel} of \cref{alg:sub}, the control policy $\mathcal{D}$ computed by MRCMC specifies at least $r$ transitions in every trajectory from the initial state of the input $G$ to its secret state $q_s$. This is because each ``sub''-policy $\mathcal{D}_j$ only specifies one transition in every trajectory from the initial state to $q_s$ because of SCT, and due to relabelling the specified transitions as unprotectable ones (i.e., uncontrollable transitions in the context of SCT), the specified transitions are unique among sub-policies. Thus, MCRMC returns a policy $\mathcal{D}$, namely a non-null value, if and only if $q_s$ is $r$-securely reachable. Note that the cost level returned by MRCMC is the maximum level among the protectable events specified by the policy $\mathcal{D}$, since $\mathcal{D}$ is constructed by taking a union of each sub-policy $\mathcal{D}_j$ for all states of $G$.

  After the iteration in lines \ref{ln:iter:mrcmc}--\ref{ln:enditer:mrcmc} in \cref{alg:drcmc}, DRCMC terminates with \texttt{NO\_SOLUTION} if it does not find any candidate policy, meaning that no secret states in $s$ are $r$-securely reachable. Since the condition \ref{cond:dssp:secure} of \cref{prob:dssp} requires that there is at least one protection policy for every pair of $R$, DRCMC returns \texttt{NO\_SOLUTION} if and only if \ref{cond:dssp:secure} of \cref{prob:dssp} does not hold. If DRCMC does not terminate at line \ref{ln:terminate}, then $\mathcal{P}^j$ in line \ref{ln:policy:protection} is a protection policy with a minimum cost among the candidates due to line \ref{ln:save}. Once all pairs of $R$ are examined, DRCMC returns the protection policies constructed by taking the union of control policies. The quantity $k$ is the maximum integer in $V$ as it is the maximum cost level of protectable events specified by the discovered policies. Thus, from MRCMC and the procedure of finding a minimum cost level in lines \ref{ln:min:start}--\ref{ln:min:end}, $k$ is minimum.
  % In other words, DRCMC does not terminate if and only if all events belonging to the cost class until $k$ are protectable.
  This implies that \ref{cond:dssp:minimum} of \cref{prob:dssp} holds if and only if DRCMC produces protection policies, i.e., not the special value \texttt{NO\_SOLUTION}. Therefore, \cref{alg:drcmc} returns a solution of \cref{prob:dssp} if and only if \cref{prob:dssp} is solvable.
\end{proof}

\cref{alg:drcmc} and \cref{thm:drcmc} are our main results of this paper. \cref{thm:drcmc} asserts that protection policies for individual agents which involve a minimum cost can always be computed by \cref{alg:drcmc} if they exist. In other words, if \cref{alg:drcmc} does not return policies and a cost level, then we can conclude that we are not able to satisfy the given security requirement.
% In the next section, we illustrate how \cref{alg:drcmc} works with the example system in \cref{exmp:motivation}.

\begin{rem}
  Since MRCMC in \cref{alg:sub} independently computes control policies for each agent, the system model is not constrained to have unique events, i.e., different agents can share the same events. Specifically, control policies for one agent computed by MRCMC are not affected by those for other agents.
\end{rem}

\begin{rem}
  The time complexity of DRCMC mainly comes from the length of nested iterations in \cref{alg:drcmc} and \cref{alg:sub}. Thus, \cref{alg:drcmc} is $O(\abs{R}nrm\abs{Q}^2)$ where $Q$ is the largest set of states among the given agents, $n$ is the total number of agents, $r$ is the largest number of required protections given by $R$, and $m$ is the largest cost level. The quantity $\abs{Q}^2$ is from the standard SCT in line \ref{ln:sub:sup} of \cref{alg:sub}, and we omit $\abs{W}$ since it always smaller than or equal to $n$, as $W$ is increased at line \ref{ln:iter:update} by at most one element in each iteration of the for-loop, and the for-loop is performed once for each of the $n$ elements of $s$. Therefore, we conclude that DRCMC is a polynomial-time algorithm.
\end{rem}

\section{Running Example}\label{sec:example}

Let us revisit the example system presented in \cref{exmp:motivation}. Since the system consists of two agents, our goal is to find two protection policies $\mathcal{P}^1$ and $\mathcal{P}^2$ for $G^1$ and $G^2$, respectively.
After the initialization in \cref{alg:drcmc}, we first have
\[ V = \emptyset \]
and
\begin{align*}
  \mathcal{P}^1(q) &= \emptyset\ \text{for all $q \in \paren{q^1_0, q^1_1, q^1_2}$} \\
  \mathcal{P}^2(q) &= \emptyset\ \text{for all $q \in \paren{q^2_0, q^2_1, q^2_2, q^2_3, q^2_4}$}.
\end{align*}
From \cref{eq:exmp:req}, the outermost routine of DRCMC begins with $(s_1, 1)$ where $s_1 = (q^1_2, q^2_2)$. Proceeding to line \ref{ln:iter:mrcmc}, DRCMC calls MRCMC with the parameters $G^1$, $s^1 = q^1_2$, and $r = 1$. In MRCMC, the subroutine \textsc{Spec} removes the secret state $q^1_2$ from $G_1 = G^1$, resulting in the specification automaton $G_{K,1}$ in \cref{fig:exmp:spec}.
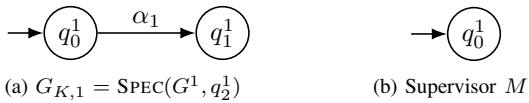
\begin{figure}[htp]
  \centering
  \begin{subfigure}{0.6\linewidth}
    \centering
    \begin{adjustbox}{scale=1}
      \begin{tikzpicture}[automata]
        \node[state,initial] (1) {$q_0^1$};
        \node[state,right=of 1] (2) {$q_1^1$};
        \draw (1) -- node{$\alpha_1$} (2);
      \end{tikzpicture}
    \end{adjustbox}
    \caption{$G_{K,1} = \textsc{Spec}(G^1, q^1_2)$}\label{fig:exmp:spec}
  \end{subfigure}
  \begin{subfigure}{0.38\linewidth}
      \centering
      \begin{adjustbox}{scale=1}
        \begin{tikzpicture}[automata]
          \node[state,initial] (1) {$q_0^1$};
        \end{tikzpicture}
      \end{adjustbox}
      \caption{Supervisor $M$}\label{fig:exmp:sup}
  \end{subfigure}
  \caption{Intermediate automata in \cref{alg:sub}}
\end{figure}

Next, with $G^1$ and $G_{K,1}$, RCMC computes a supervisor and finds a minimum value of $k$ in the subroutine. From $\Sigma_{cl,1} = \paren{\beta_1, \beta_4}$, a resulting supervisor $M$ with $k = 1$ is indeed empty, thus RCMC increments $k$ to $2$. Since $\Sigma_{p,2} = \paren{\alpha_1, \beta_1, \beta_2, \beta_4}$ that contains $\alpha_1$, RCMC returns a pair of $k = 2$ and the nonempty supervisor $M$ in \cref{fig:exmp:sup}.
Going back to MRCMC, we can derive from $M$ the following control policy $\mathcal{D}_1$:
\begin{align*}
  \mathcal{D}_1(q^1_0) &= \paren{\alpha_1} \\
  \mathcal{D}_1(q) &= \emptyset\ \text{for all $q \in \paren{q^1_1, q^1_2}$}.
\end{align*}
Since $j = r = 1$ within the first iteration, the loop terminates and MRCMC returns the control policy $\mathcal{D}^1 = \mathcal{D}_1$ and the maximum cost level $\max(T) = 2$ found in RCMC. As MCRMR returned a nonempty supervisor, we now have
\[ W = \paren{(\mathcal{D}^1, 2)}. \]

The procedure MRCMC is called twice in every iteration over $R$ (i.e., the lines between \ref{ln:main:start} and \ref{ln:main:end} of \cref{alg:drcmc}), since we have that $n = 2$ in this example. For $G^2$ and $s^2 = q^2_2$, MRCMC returns $(\mathcal{D}^2, k^2)$ where $k^2 = 1$ and
\begin{subequations}\label{eq:policy:s1}
  \begin{align}
    \mathcal{D}^2(q^2_0) &= \paren{\beta_1, \beta_4}\\
    \mathcal{D}^2(q) &= \emptyset\ \text{for all $q \in \paren{q^2_1, q^2_2, q^2_3, q^2_4}$}.
  \end{align}
\end{subequations}
Thus, we have
\[ W = \paren{(D^1, 2), (D^2, 1)}. \]
Since $W$ is not empty, DRCMC proceeds to find a policy associated with the minimum cost level. Since $k^2 = 1$ in a pair of $(D^2, 1)$ is minimum in the set of computed policies $W$, DRCMC merges $D^2$ into $\mathcal{P}^2$ and saves $k^2$ into $V$, discarding $D^1$ computed in this iteration.

For $(s_2, 2)$ where $s_2 = (q^1_2, q^2_4)$, MRCMC calls RCMC twice in total since it has $r = 2$.
The procedure MRCMC first computes the policies and cost level to protect $q^1_2$ that result in
\begin{subequations}\label{eq:policy:s2}
  \begin{align}
    \mathcal{D}^1(q^1_0) &= \paren{\alpha_1} \\
    \mathcal{D}^1(q^1_1) &= \paren{\alpha_2} \\
    \mathcal{D}^1(q^1_2) &= \emptyset
  \end{align}
\end{subequations}
and
\[ k^1 = 3. \]
For $q^2_4$, on the other hand, RCMC in the first iteration returns the control policy
\begin{align*}
  \mathcal{D}_1(q^2_0) &= \paren{\beta_1, \beta_4} \\
  \mathcal{D}_1(q) &= \emptyset\ \text{for all $q \in \paren{q^2_1, q^2_2, q^2_3, q^2_4}$}
\end{align*}
which is by chance the same as that in \cref{eq:policy:s1}. By line \ref{ln:mrcmc:relabel} of MRCMC, the transitions with $\beta_1$ and $\beta_4$ at state $q^2_0$ are relabelled to $\beta_1'$ and $\beta_4'$ which are unprotectable, resulting in the relabelled intermediate agent $G_2$ in \cref{fig:exmp:relabelled:2}.
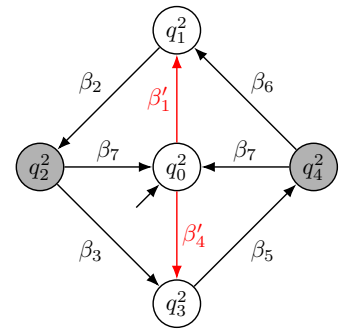
\begin{figure}[htp]
  \centering
  \begin{adjustbox}{scale=0.9}
    \begin{tikzpicture}[automata]
      \node[state] (1) {$q_0^2$};
      \node[state,above=of 1] (2) {$q_1^2$};
      \node[state,left=of 1,secret] (3) {$q_2^2$};
      \node[state,below=of 1] (4) {$q_3^2$};
      \node[state,right=of 1,secret] (5) {$q_4^2$};

      \draw (-6mm,-6mm) -- (1);
      \draw[red] (1) -- node{$\beta_1'$} (2);
      \draw (2) -- node[swap]{$\beta_2$} (3);
      \draw (3) -- node[swap]{$\beta_3$} (4);
      \draw[red] (1) -- node{$\beta_4'$} (4);
      \draw (4) -- node[swap]{$\beta_5$} (5);
      \draw (5) -- node[swap]{$\beta_6$} (2);
      \draw (3) -- node[above]{$\beta_7$} (1);
      \draw (5) -- node[above]{$\beta_7$} (1);
    \end{tikzpicture}
  \end{adjustbox}
  \caption{Relabelled intermediate agent $G_2$ (the same as that in \cref{fig:exmp:relabelled}); red arrows indicate relabelled transitions}\label{fig:exmp:relabelled:2}
\end{figure}
Although all transitions in $G^2$ are labelled by unique events, it is worth noting again that only specified transitions (not events) by the control policy are relabelled, thus one agent can have original and relabelled versions of the same event.

In the next iteration, RCMC returns the intermediate pair of control policy and cost level $(\mathcal{D}_2, k_2)$ where $k_2 = 4$ and
\begin{subequations} \label{eq:policy:s2:n}
  \begin{align}
    \mathcal{D}_2(q^2_3) &= \paren{\beta_5} \\
    \mathcal{D}_2(q) &= \emptyset\ \text{for all $q \in \paren{q^2_0, q^2_1, q^2_2, q^2_4}$}.
  \end{align}
\end{subequations}
The event $\beta_2$ at state $q^2_1$ was not specified due to the characteristics of SCT, because disabling $\beta_2$ at $q^2_1$ would not prevent the other secret state $q_4^2$ from being reached, given that the set of protectable events for the supervisor $M$ here is $\Sigma_{p,4} = \paren{\alpha_1, \alpha_2, \beta_1, \beta_2, \beta_4, \beta_5}$.
Therefore, the supervisor $M$ specifies the disablement of $\beta_5$ at state $q_3^2$ which then ensures that $q_4^2$ will not be reachable.
At this point, DRCMC finds that $\mathcal{D}^1$ in \cref{eq:policy:s2} is associated with the smaller cost level than $\mathcal{D}^2$ in \cref{eq:policy:s2:n}. Thus, DRCMC selects $\mathcal{D}^1$ to merge into $\mathcal{P}^1$ for $s_2$.

Finally, from the control policies and cost levels computed in the main loop in lines \ref{ln:main:start}--\ref{ln:main:end}, DRCMC returns a solution to \cref{prob:dssp} for \cref{exmp:motivation} that is given by:
\begin{subequations}\label{eq:solution:p1}
  \begin{align}
    \mathcal{P}^1(q^1_0) &= \paren{\alpha_1} \\
    \mathcal{P}^1(q^1_1) &= \paren{\alpha_2} \\
    \mathcal{P}^1(q^1_2) &= \emptyset,
  \end{align}
\end{subequations}
\begin{subequations}\label{eq:solution:p2}
  \begin{align}
    \mathcal{P}^2(q^2_0) &= \paren{\beta_1, \beta_4} \\
    \mathcal{P}^2(q) &= \emptyset\ \text{for all $q \in \paren{q^2_1, q^2_2, q^2_3, q^2_4}$}
  \end{align}
\end{subequations}
and
\[k = 3.\]
Interestingly, $\mathcal{P}^1$ in \cref{eq:solution:p1}, in fact, solely satisfies the security requirements of both $s_1$ and $s_2$, even though $\mathcal{P}^2$ computed by \cref{alg:drcmc} is not empty. Thus, one may claim that the solution of \cref{prob:dssp} given by $\mathcal{P}^1$ and $\mathcal{P}^2$ in \cref{eq:solution:p1} and \cref{eq:solution:p2}, respectively, is suboptimal. Indeed, that is not always the case. From the security perspective, we consider that protecting more transitions than needed (e.g., keeping $\mathcal{P}^2$ nonempty after we discovered that $\mathcal{P}^1$ solely satisfies the security requirements) is still acceptable as long as a protection policy computed by \cref{alg:drcmc} satisfies the minimality of cost level, namely the condition \ref{cond:dssp:minimum} of \cref{prob:dssp}.

\section{Conclusion}\label{sec:conclusion}

We have presented the distributed secret securing problem (DSSP) which extends previous work from monolithic systems to multiagent ones. In DSSP, we modelled the global secret as the set of tuples of secret states from local agents. Our goal was to find a protection policy such that at least one secret state of every tuple in the global secret is protected by at least the required number of protections, and the protection cost is minimum. We characterized the solvability of DSSP by deriving a necessary and sufficient condition, and showed that our algorithm developed in this paper produces a solution of DSSP when the condition holds. Finally, that algorithm was illustrated using the example system consisting of distributed databases that store passwords.

In future work, we aim to extend our representation of protections from events to transitions, namely, the situation where there are protectable transitions instead of protectable events in the system. Another possible direction is to implement the methodology of this paper in practical distributed systems.

\bibliographystyle{IEEEtran}
\bibliography{IEEEabrv,references}
\end{document}